\documentclass[12pt]{amsart}
\usepackage{amssymb}
\usepackage{amsmath}
\newtheorem{theorem}{Theorem}
\newtheorem{lemma}{Lemma}

\usepackage{graphicx}

\newtheorem*{theorem*}{Theorem}{\bf}{\it}
\newtheorem*{proposition*}{Proposition}{\bf}{\it}

\begin{document}

\title[Morgul T=2]{Further finding results for cycles by semilinear control}

\author{Dmitriy Dmitrishin, Elena Franzheva and Alex Stokolos}

\begin{abstract} \"O. Morg\"ul \cite{9} has suggested the following modification of the scalar dynamical system $x_{n+1}=f(x_n)$ to find periodic solutions of the period $T$
$$
x_{n+1}=(1-\gamma)f(x_n)+\gamma x_{n-T+1},\quad 0<\gamma<1.
$$ 
The schedule was extended to the vector case by Dmitirshin et all in \cite{Har}.  The advantage of the proposal scheme is in its simplicity. The main goal of the current paper is to investigate two-delays generalization in the form
$$
x_{n+1}=(1-\gamma)(a_1f(x_n)+a_2f(x_{n-T}))+\gamma (b_1x_{n-T+1}+b_2x_{n-2T+1}).
$$
The provided analysis suggests the effective choice of coefficients
$$
a_1=\frac{T+1}{T+2},\; a_2=\frac1{T+2},\quad b_1=b_2=\frac12,\quad \gamma\in\left(0,\frac{1.13}{T-1}\right], T>1.
$$
 The modified schedule allows to stabilize the cycles with dispersion of multipliers three times larger compare to Morg\"ul's scheme. The proof heavily involves methods of real analysis and geometric complex function theory.

\end{abstract}
\maketitle

\section{Introduction}

Accordingly to \cite{9} "Chaotic behavior is a very interesting and fascinating phenomenon which is frequently observed in many physical systems." As mathematical models for the description of chaotic behavior are used discrete dynamical systems. By control of chaos we mean a small external influence on the system or a small change in the structure of the system in order to transform the chaotic behavior of the system into regular (or chaotic, but with different properties) \cite{1}. The problem of optimal influence on the chaotic regime is one of the fundamental problems in nonlinear dynamics \cite{2,3}.

It is assumed that the dynamic system has a chaotic attractor, which contains a countable set of unstable cycles of different periods. If the control action locally stabilizes a cycle, then the trajectory of the system remains in its neighborhood, i.e. regular movements will be observed in the system. Hence, one of the way to control chaos is the local stabilization of certain orbits from a chaotic attractor.

To solve the stabilization problem, various control schemes were proposed \cite{4}, among them the controls based on the Delayed Feedback Control (DFC) principle are quite popular \cite{5}. Such controls, under certain conditions, allow local stabilization of equilibrium positions or cycles, which, generally speaking, are not known in advance. Among the DFC schemes, linear schemes are the simplest for physical implementation. However, they have significant limitations: they can be used only for a narrow area of the parameter space that enter the original nonlinear system, c.f. \cite{arxiv}.

To extend the class of systems to which the DFC scheme applies, it is necessary to introduce non-linear elements into the control. For the first time, a nonlinear DFC with one delay was considered in \cite{6}, where the advantages of such a modification are also noted, in particular, the fact that the control becomes robust. In \cite{7,8} the concept of nonlinear control with one delay from \cite{6} was extended: to the vector case; to the case of several delays; to the case of an arbitrary period $T$. It is shown that the control allows to stabilize cycles of arbitrary lengths, unless the multipliers are real and greater than one. A relationship is established between the size of the localization set of multipliers and the amount of delay in the nonlinear feedback.

In \cite{9,10}, a semilinear DFC scheme with linear and nonlinear elements was investigated. In spite of the fact that this scheme contains only one difference in control, nevertheless, it is possible to stabilize cycles with length $T = 1, 2$ under sufficiently general assumptions about cycle multipliers. For $T\ge 3$, the situation changes critically, and the stabilization of cycles is possible only if the rigid constraints on multipliers are met, namely only if cycle multiplier is in $(-(T/(T-2))^T,1).$ Our goal is to suggest a simple schedule finding/stabilizing cycles with complex multipliers lie in a region significantly wider then in the Morgul's case. 

In particular, the case $T=1,2$ is done in \cite{Har} for all N. In our situatuion that implieas the chooice $a_1=b_1=2/3$ and $a_2=b_2=1/3$ for $T=1$ and the choice of $a_1=3/4, a_2=1/4$ and $b_1=2/3, b_2=1/3$ for $T=2.$

Note, the suggested semilinear control is a convex combination of linear and non-linear controls. This blend has a very interesting property - the linear control is bad, the non-linear is better while the combined control produced fascinating effect.  

The non-linear control is controlling the situation while the linear part increases the rate of convergence and enlarges the possible area of localization of multipliers. 

\section{Review and preliminary results}

We consider vector nonlinear discrete system, which in the absence of control has the form
\begin{equation}\label{(1)}
{x_{n + 1}} = f\left({x_n}\right), \quad {x_n} \in {\mathbb R^m}, \quad n = 1, 2, \ldots,
\end{equation}
where $f\left(x\right)$ is a differentiable vector function of the corresponding dimension. It is assumed that the system (1) has an invariant convex set $A$, that is, if $\xi \in A$, then $f\left(\xi\right) \in A$. It is also assumed that in this system there is one or more unstable $T$-cycles $\left({{\eta_1}, \ldots, {\eta_T}}\right)$, where all the vectors ${\eta_1}, \ldots, {\eta_T}$ are distinct and belong to the invariant set $A$, i.e. ${\eta_{j + 1}} = f\left({{\eta_j}}\right), j = 1, \ldots, T - 1, {\eta_1} = f\left({{\eta_T}}\right)$. The multipliers of the unstable cycles under consideration are defined as the eigenvalues of the products of the Jacobian matrices $\prod_{j = 1}^T {{{f'}_{}}\left({{\eta_j}}\right)}$ of dimensions $m \times m$. As a rule, the cycles $\left({{\eta_1}, \ldots, {\eta_T}}\right)$ of the system (1) are not a priori known as well as the spectrum  $\left\{ {{\mu_1}, \ldots, {\mu_m} } \right\}$ of the matrix $\prod_{j = 1}^T {f'\left({{\eta_j}}\right)}$.

 To stabilize the cycle of the length $T$ Morg\"ul \cite{9,10}
proposed a feedback  control that includes linear and nonlinear elements, i.e., a semilinear feedback control of the form
\begin{equation}\label{(6)}
{u_n} = - \gamma \left({f({x_n}) - {x_{n - T + 1}}}\right),
\end{equation}
for which a corresponding closed-loop system is
\begin{equation}\label{(7)}
x_{n+1}=(1-\gamma) f(x_n) + \gamma x_{n-T+1}, 
\end{equation}
where $\gamma \in \left[ {0, 1}\right)$. On the cycle, the conditions $f({x_n}) = {x_{n + 1}} = {x_{n - T + 1}}$ are fulfilled, therefore, on the cycle ${u_n} \equiv 0$.

Note that in \cite{9} only the scalar case $f: \mathbb{R} \to \mathbb{R}$ was considered. In that case the stability is possible if multiplier (unique and real) belongs to the interval $(-(T/(T-2))^T,1).$ It was shown in \cite{Har} that the stabilization is possible if and only if all multipliers (real as well as complex) are in the region 
$$
W=\left\{w=\left(\frac{T-1}{T-2}\right)^T\frac{(1-\frac z{T-1})^T}z: \; z\in\mathbb D\right\},\qquad \mathbb D=\{z:|z|<1\}.
$$
Our goal is to expand the region $W.$ To do that we suggests to consider a control
$$
u=-(1-\gamma)\epsilon(f(x_n)-f(x_{n-T}))-\gamma\left(\delta_1(f(x_n)-x_{n-T+1})+\delta_2(f(x_{n-T})-x_{n-2T+1})\right),
$$
where $\epsilon\ge 0,$ $\delta_j\ge 0$ and $\delta_1+\delta_2=1.$ Note that in this case on the $T$-cycle  $u\equiv 0$ as it was in Morg\"ul's case. And the close-loop system is
$$
x_{n+1}=(1-\gamma)(a_1f(x_n)+a_2f(x_{n-T}))+\gamma(b_1x_{n-T+1}+b_2x_{n-2T+1}),
$$
where $a_1=\frac{1-\gamma\delta_1}{1-\gamma}-\epsilon,$ 
$a_2=\epsilon-\frac{\gamma\delta_2}{1-\gamma}$ and $\delta_1=b_1,\delta_2=b_2.$
It is clear that $a_1+a_2=b_1+b_2=1.$

If $x_n$ is scalar and $\gamma=0$ then the problem was considered in \cite{6} and if $\gamma=1$ and $\delta_2=0$ it was considered in \cite{9}. 

As well known the investigation of the local stability of cycles is reduced to the verification that the roots of the characteristic polynomials belong to the unit disc in the complex plane (c.f. \cite{14}). In our case the characteristic polynomial is the following (see \cite{Har})
\begin{equation}\label{(9)}
\tilde P(\lambda)=\prod_{j = 1}^m {\left[ \lambda^{2T} {\left({1-\gamma p(\lambda^{-1})}\right){ ^T} - {\mu_j} {{\left({1 - \gamma }\right)}^T} {\lambda ^{2T - 1}}}(q(\lambda^{-1}))^T \right]}, 
\end{equation}
where $\mu_j$ are multipliers of the cycle ($j = 1, \ldots, m$), in general, complex and the polynomials $p$ and $q$ are $p(\lambda^{-1})=b_1\lambda^{-1}+{b_2}\lambda^{-2},$ $q(\lambda^{-1})=a_1+{a_2}\lambda^{-1}.$ 

If the multipliers $\mu { _j}$, $j = 1, \ldots, m$ are determined exactly, then one can check whether the roots belong to the central unit disc by known criteria such as Schur-Cohn, Clark, Jury \cite{14}. However, cycles are not known, hence, multipliers are not known.

In this case, the geometric criterion of A. Solyanik proved to be effective for the stability of cycles of discrete systems \cite{15}. Let us apply this criterion.

Let $\lambda = \frac{1}{z}$ and 
$$\Phi \left(z\right) = {\left({1 - \gamma }\right)^T} \frac{z(q(z))^T}{{{{\left({1 - \gamma p(z) }\right)}^T}}}.$$ 

Then the following lemma is valid (c.f. \cite{HA1}).

\begin{lemma}\label{lemma1}
The polynomial \eqref{(9)} is Schur stable if and only if 
\begin{equation}\label{(10)}
\mu_j\in \left(\overline{\mathbb{C}}\backslash\Phi(\overline{D})\right)^*, \quad j = 1, \ldots, m,
\end{equation}
where $\overline{D} = \left\{z \in \mathbb{C}: \left|z\right|\le 1\right\}$ is a closed central unit disk, $\overline{\mathbb{C}}$ is an extended complex plane, the asterisk denotes the inversion  ${\left(z\right)^*} = \frac{1}{{\bar z}}$. Here $\bar z$ denotes the complex conjugated of $z$.
\end{lemma}

Let $m$ be a set of localization of multipliers and $M^*$ is inversion to $M.$ It follows from the Lemma 1 that the polynomial is Schur stable if  the $M^*$ is an exceptional set for the image of $\mathbb D$ under the map $\Phi(z).$ This property will be central in the construction of the control coefficients $a_1, a_2,$ $b_1, b_2$ and averaging parameter $\gamma.$

Let us outline a geometric meaning of the Lemma. Namely, in case of no control the system has a form \eqref{(1)}. Thus $\Phi(z)=z.$ Then $\Phi(\bar{\mathbb D})=\bar{\mathbb D}$ and 
$(\mathbb C\backslash \bar{\mathbb D})^*=\mathbb D.$ Thus, we came up with the standard stability condition $\mu_j\in \mathbb D.$ In other words, geometrically the addition of the control means extension of the central unit disc to a set that cover all multipliers.

\section{Conception of further strategy}
Thus, our goal to make the set 
$$
W=\left\{w=\frac{(1-\gamma(b_1z+b_2z^2))^T}{(1-\gamma)^Tz(a_1+a_2z)^T}: \;|z|<1\right\}
$$ as large as possible by choosing control coefficients. In future we will control the linear size of the set along the real line. I.e. we want the set $W$ to cover the interval 
$(-\mu^*,1)$ of the largest length. If the function $1/\Phi(z)$ is typically real in $\mathbb D$ then 
$$
\mu^*=\left(\frac{1-\gamma(b_2-b_1)}{(1-\gamma)(a_1-a_2)}
\right)^{T}$$
In case of $\Phi(z)$ is polynomial one can apply the theory of the polynomials of maximal range by S.Rushenweih \cite{Ru}. By Rushenweih's theory the extremal polynomials should be univalent and zeros of the derivative should lie on the unit disc. However, even for the polynomials the application of this theory is a very tricky task. In the case of stabilization of $T$-cycles these polynomials should be $T$-symmetric. A various sets of polynomials for cycle stabilization can be found in \cite{HA,HA1}.

For the semilinear control it is necessary to construct not polynomials rather rational functions. For those functions the Rushenweih's theory does not have analogy. This is why the consideration of the case $N=2$ is of independent interest.

First of all, the coefficients $a_1$ and $a_2$ can be used as in the non-linear control, namely
$a_1=\frac{T+1}{T+2}$ and $a_2=\frac{1}{T+2}.$

Second, we would like to chose the coefficients $b_1, b_2$ such that the denominator in the formula for $\mu^*$ would be as small as possible, thus $\gamma$ as large as possible to increase the size of the set $W.$

This is why the properties of typically real is a main property. Further analysis indicates that the property of univalency is difficult to verify for rational functions even in simple cases. However, in addition to be typical real our function $\Phi(z)$ will be at same time locally univalent. We will see that the property of typically real makes the formulas much more complicated however improves the size of the region incrementally. In case of local univalency we get
$\mu^*=\left(\frac{T+2}{(1-\gamma) T}\right)^T.$ What left is to find maximal value for $\gamma.$ Note that the expected value $\frac1{T-1}$ is not maximal.

\section{Typically realness}
Our goal is to determine conditions for the function
	\begin{equation}\label{typ}
\Phi(z)=(1-\gamma)^T\frac{z(\frac{T+1}{T+2}+\frac z{T+2})^T}{{{{\left({1 - \gamma(\alpha z+(1-\alpha)z^2) }\right)}^T}}}
\end{equation}
to be typically real in the disc $\mathbb D.$

Let us remind that the function is typically real in $\mathbb D$ if it is real at every real point of the disc and in all the others points of the disc we have
$$
\Im\{f(z)\}\Im\{z\}>0.
$$

We consider the class of functions $f(z), f(0)=0$ that are analytic and typically real in $\mathbb D.$ 

It is clear that  if $f(z)$ is typically real in $\mathbb D$ then its T-symmetrized version $\phi(z)={}^T\hspace{-.2cm}\sqrt{f(z^T)}$ is
typically real in $\mathbb D$. The converse is true as well.

Let us consider a general function with arbitrary polynomials $\tilde q(z), \tilde p(z)$
\begin{equation}\label{Fz}
\tilde\Phi(z)=(1-\gamma)^T\frac{z[\tilde q(z)]^T}{(1-\gamma \tilde p(z))^T},\quad \tilde q(1)=1, \tilde p(0)=0, \tilde p(1)=1
\end{equation}
and assume that $z[\tilde q(z)]^T$ is typically real polynomial.

This function corresponds the associated T-symmetric function 
\begin{equation}\label{phiz}
\tilde F(z)=(1-\gamma)\frac{z[\tilde q(z^T)]}{1-\gamma \tilde p(z^T)}
\end{equation}

Let us define the real functions
$$
C_1(t)=\frac{\Im\{\tilde q(-e^{iTt})e^{it}\}}{\Im\{\tilde  q(-e^{iTt})\tilde p(-e^{-iTt})e^{it}\}}
$$
and
$$
C_2(t)=\frac{\Im\{\tilde  q(e^{iTt})e^{it}\}}{\Im\{\tilde  q(e^{iTt}) \tilde p(e^{-iTt})e^{it}\}}.
$$
on the the sets 
$$
\tau_i=\{t: C_i(t)>0\}\cap (0,\frac\pi T),\qquad i=1,2.
$$
The function $C_1$ (or $C_2$ ) is undefined if the denominator is identically zero. In that case we assume that $\tau_1=\emptyset$ (or $\tau_2=\emptyset.$ )\newpage

\begin{theorem}
Let 
$$
c_1^*=\inf_{t\in\tau_1}C_1(t),\qquad
c_2^*=\inf_{t\in\tau_2}C_2(t),\qquad c^*=\min\{c_1^*,c_2^*\}.
$$

i) If $\tau_1=\tau_2=\emptyset$, then the function $\tilde F(z)$ is typically real for all $\gamma>0.$

ii) If $c^*=0$ then the function $\tilde F(z)$ is typically real for no $\gamma>0.$

iii) If $c^*>0$ then the function $\tilde F(z)$ is typically real for $\gamma\in[0,c^*).$
\end{theorem}

Note that $C_1(\frac\pi T)=1$ implies that $c^*\le 1$ if exists.

\begin{proof} 
The function $\tilde F(z)$ is not typically real if there exist parameters $(\gamma_1,\rho_1,t_1)\in[0,1]\times(0,+\infty)\times (0,\frac\pi T)$  such that
\begin{equation}\label{A}
\frac{e^{it_1}\tilde q(e^{iTt_1})}{1-\gamma_1\tilde p(e^{iTt_1})}=\rho_1 e^{i\frac{\pi}T},
\end{equation}
or there exist parameters $(\gamma_2,\rho_2,t_2)$  such that
\begin{equation}\label{B}
\frac{e^{it_2}\tilde q(e^{iTt_2})}{1-\gamma_2 \tilde p(e^{iTt_2})}=\rho_2.
\end{equation}

The relation \eqref{A} implies that
$$
\gamma_1\tilde p(e^{iTt_1})+\frac1{\rho_1} e^{-i\frac\pi T} e^{it_1}\tilde q(e^{iTt_1})=1
$$
$$
\gamma_1 \tilde p(e^{-iTt_1})+\frac1{\rho_1} e^{i\frac\pi T} e^{-it_1}\tilde q(e^{-iTt_1})=1
$$
and therefore
$$
\gamma_1=\frac{\tilde q(e^{iTt_1})e^{-i\frac\pi T}e^{it_1} -\tilde q(e^{-iTt_1})e^{i\frac\pi T}e^{-it_1}}
{\tilde p(e^{-iTt_1}) \tilde q(e^{iTt_1})e^{-i\frac\pi T}e^{it_1} -\tilde p(e^{iTt_1})\tilde q(e^{-iTt_1})e^{i\frac\pi T}e^{-it_1}}=
$$
$$
\frac{\Im\{\tilde q(e^{iTt_1})e^{-i\frac\pi T}e^{it_1}\}}{\Im\{\tilde p(e^{-iTt_1})\tilde q(e^{iTt_1})e^{-i\frac\pi T}e^{it_1}\}}
$$
Let us introduce the function
$$
C_1(t)=\frac{\Im\{\tilde q(e^{iTt})e^{-i\frac\pi T}e^{it}\}}{\Im\{\tilde p(e^{-iTt})\tilde q(e^{iTt})e^{-i\frac\pi T}e^{it}\}}
$$

Similarly, the condition \eqref{B} implies 
$$
\gamma_2=\frac{\Im\{\tilde q(e^{iTt_2})e^{it_2}\}}{\Im\{\tilde p(e^{-iTt_2})\tilde q(e^{iTt_2})e^{it_2}\}}
$$
and consider the function
$$
C_2(t)=\frac{\Im\{\tilde q(e^{iTt})e^{it}\}}{\Im\{\tilde p(e^{-iTt})\tilde q(e^{iTt})e^{it}\}}
$$ 

The function $\tilde F(z)$ is typically real for some real $\gamma$ if both equations $C_j(t)=\gamma, j=1,2$ do not have solutions on the set $(0,\frac\pi T).$\\

The function $z[\tilde q(z)]^T$ is typically real therefore the functions  $z[\tilde q(-z)]^T,$ $z\tilde q(z^T),$ and $z\tilde q(-z^T)$
are typically real as well.\\

Due to the relations
$$
\frac{\Im\{\tilde q(e^{iTt})e^{-i\frac\pi T}e^{it}\}}{\Im\{\tilde p(e^{-iTt})\tilde q(e^{iTt})e^{-i\frac\pi T}e^{it}\}}=
\frac{\Im\{\tilde q(e^{-iTt})e^{i\frac\pi T}e^{-it}\}}{\Im\{\tilde p(e^{iTt})\tilde q(e^{-iTt})e^{i\frac\pi T}e^{-it}\}}
$$
and after substitution $\xi=\frac\pi T-t,\; \xi\in(0,\frac\pi T)$ one gets
$$
\frac{\Im\{\tilde q(-e^{iT\xi})e^{i\xi}\}}{\Im\{\tilde p(-e^{-iT\xi})\tilde q(-e^{iT\xi})e^{i\xi}\}}.
$$
Therefore, the set of the values of the functions $C_1(t)$ coincides with the set of the values of the function
$$
\frac{\Im\{\tilde q(-e^{iTt})e^{it}\}}{\Im\{\tilde p(-e^{-iTt})\tilde q(-e^{iTt})e^{it}\}}
$$
on the set $t\in(0,\frac\pi T).$\\

The function $\Im\{\tilde q(-e^{iTt})e^{it}\}>0$ for  $t\in(0,\frac\pi T).$ That means that the graph of the function $C_1(t)$ does not intersect the $X$-axis, however can touch it.

The function $C_1(t)$ either continuous on $(0,\frac\pi T)$ or has discontinuity of the second type. This means that  the set of values of the function $C_1(t)$ for $t\in\tau_1$ is an interval $(c_1^*,\beta_1)$ with $c_1^*\le 1$ and $\beta_1$ is either a number or infinity.

In the same way, the function $C_2(t)$ maps the set $\tau_2$ on the interval $(c_2^*,\beta_2)$  and $\beta_2$ is either a number or infinity.

Define now $c^*=\min\{c_1^*,c_2^*\}.$ If this number exists and $c^*=0,$ then the conclusion ii) of the theorem is proved. If $c^*>0,$ then the conclusion iii) of the theorem is proved. If $c^*$ does not exists, then the conclusion i) is valid.

\end{proof}

The subtlety of the situation is in the discontinuity of typically real and univalent properties on parameters in general. However, in our case the set of $\gamma$ where it is typically real and contains zero is connected and contains zero. I.e. if for some $\tilde\gamma$ the function $\tilde \Phi(z)$ is typically real then it stays typically real for all $\gamma\in[0,\tilde\gamma].$\\

The other remark in order is that to restrict ourselves with only function $C_1(t)$ is wrong as demonstrates the example of the function \eqref{typ} with $T=7, \alpha=-0.6$

\section{Examples}

\subsection{Example 1} Let $q(z)\equiv 1,\, p(z)=z.$ For $T=1$ the functions $C_1(t)$ and $C_2(t)$ are undefined. That means that the function
$F(z)=z/(1-\gamma z)$ is typically real in $\mathbb D$ for all $\gamma$ (possibly with the exception of the pole).

For $T\ge2$ we have
$$
C_1(t)=\frac{\Im\{e^{it}\}}{\Im\{-e^{-itT}e^{it}\}}=\frac{\sin t}{\sin(T-1)t},
$$
$$
C_2(t)=\frac{\Im\{e^{it}\}}{\Im\{e^{-itT}e^{it}\}}=-\frac{\sin t}{\sin(T-1)t}.
$$
Note that $C_2(t)<0$ on $(0,\frac\pi T)$ and that the function $C_1(t)$ is increasing, therefore 
$$
c_1^*=\lim_{t\to 0} \frac{\sin t}{\sin(T-1)t}=\frac 1{T-1}.
$$
Thus, the function $F(z)=z/(1-\gamma z)^T$ is typically real for $\gamma\in [0,\frac1{T-1}].$

\subsection{Example 2} Let $q(z)\equiv 1,\, p(z)=z^2.$ In this case
$$
C_1(t)=\frac{\Im\{e^{it}\}}{\Im\{e^{-i2Tt}e^{it}\}}=-\frac{\sin t}{\sin(2T-1)t},\qquad C_2(t)=C_1(t).
$$
For $T=1$ the function $F(z)=z/(1-\gamma z^2)$ is typically real for $\gamma\in[0,1]$ because $C_1(t)\equiv -1<0.$

For $T\ge 2$ we have $\tau_1=\left(\frac\pi{2T-1},\frac\pi{T}\right).$ Numerically we can find $c_1^*$\\

\centerline{
\begin{tabular}{|p{0.15in}||p{0.15in}|p{0.15in}|p{0.32in}|p{0.32in}|p{0.32in}|p{0.32in}|p{0.32in}|p{0.32in}|p{0.32in}|} \hline 
$T$ & 2  & 3  & 4 & 5 & 6 & 7 & 8 & 9 & 10 \\ \hline 
$c_{1}^{*} $ & 1 & 0.8 & 0.613 & 0.490 & 0.407 & 0.347 & 0.302 & 0.268 & 0.240 \\ \hline 
\end{tabular}
}
\medskip

The function $F(z)=z/(1-\gamma z^2)^T$ is typically real for $\gamma\in [0,c_1^*(T)].$

\subsection{Example 3} Let $q(z)=\frac{T+1}{T+2}+\frac1{T+2}z, \, p(z)=\alpha z+(1-\alpha)z^2,\; \alpha\in[0,1].$ In this case
$$
C_1(t)=\frac{(T+1)\sin t-\sin(T+1)t}{\alpha\sin t + (\alpha T+1)\sin(T-1)t -(T+1)(1-\alpha)\sin(2T-1)t}
$$
and
$$
C_2(t)=\frac{(T+1)\sin t+\sin(T+1)t}{\alpha\sin t - (aT+1)\sin(T-1)t -(T+1)(1-\alpha)\sin(2T-1)t}
$$
Let us consider the behavior of the function $C_1(t)$ in the neighborhood of zero
$$
C_1(t)=\frac{\frac16 T(T+1)(T+2)t^3+o(t^4)}{T^2(3\alpha-2)t+T[T^3(\frac43-\frac32\alpha)+T^2(-\frac56+\frac76\alpha)+T(-\frac12+\frac12\alpha)+(\frac13 -\frac23\alpha)]t^3+o(t^4)}
$$
If $\alpha>\frac23$ then in the neighborhood of zero $C_1(t)>0$ and $\lim_{t\to 0}C_1(t)=0.$ Therefore for $\alpha>\frac23$ the function \eqref{typ}
won't be typically real for any $\gamma\not=0.$

If $a=\frac23$ then the function $C_1(t)$ is increasing on $(0,\frac\pi T)$ and
$$
\lim_{t\to 0}C_1(t)=\frac{3(T+1)(T+2)}{(T-1)(6T^2+5T+2)}. 
$$
Therefore
$$
c_1^*=\frac{3(T+1)(T+2)}{(T-1)(6T^2+5T+2)}. 
$$
The function $C_2(t)<0$ for all $T>1$ and  $t\in(0,\frac\pi T).$ For $T=1$ the function $C_2(t)$ is undefined. 

Thus the function \eqref{typ} is typically real for $\gamma\in[0,c_1^*].$\\

\medskip

Now, let $a<\frac23.$ For each such $a$ and for $T=3,4,...$ the boundary for the parameter $\gamma$ for which the function $F(z)$ is typically real can be determined numerically, using the theorem. 

Let us show several first values of $c^*(T)$ for $\alpha=\frac12.$\\

\centerline{
\begin{tabular}{|p{0.15in}||p{0.32in}|p{0.32in}|p{0.32in}|p{0.32in}|p{0.32in}|p{0.32in}|p{0.32in}|p{0.32in}|p{0.32in}|} \hline 
$T$ & 3 & 4 & 5 & 6 & 7 & 8 & 9 & 10\\ \hline 
$c_{1}^{*} $ & .6415 & .4364 & 0.3247&.2570  & .2119 &.1800  & .1561 & .1378   \\ \hline 
\end{tabular}
}
 
 \section{Local univalence}

In this section we impose an additional requirement of the local univalncy on the function \eqref{typ} and compare how much change appears for the region
$$
W=\frac1{\Phi(\mathbb D)}.
$$
 
 Recall that the function $f(z)$ is locally univalent in the unit disc $\mathbb D$ if $f^\prime(z)\not=0$ in $\mathbb D.$ Since the condition $\alpha\le\frac23$ is necessary for the typicall reallness we will consider only that choice of $\alpha$.
\newpage

\begin{theorem}
A function \eqref{typ} is locally univalent in $\mathbb D$ if and only if $0\le\gamma\le\gamma_0$ where
$$
\gamma_0=\frac12\left( \frac{(T+1)((2T^2+1)(1-\alpha)-1)}{(T-1)^2(1-\alpha)(\alpha(T+2)-1)} -\sqrt{\mathcal D}\right),
\quad\alpha>\frac1{T+2}.
$$
$$
\gamma_0= \frac{(T+1)(T+2)}{(2T^2+1)(T+1)-(T+2)},\qquad \alpha=\frac1{T+2}.
$$
$$
\gamma_0=\frac12\left( \frac{(T+1)((2T^2+1)(1-\alpha)-1)}{(T-1)^2(1-\alpha)(\alpha(T+2)-1)} +\sqrt{\mathcal D}\right),
\qquad \alpha<\frac1{T+2}.
$$
Above
$$
\mathcal D=\frac{(T+1)^2\left(((2T^2+1)(1-\alpha )-1)^2-4(T-1)^2(1-\alpha)(\alpha(T+2)-1)\right)}{\left[(T-1)^2(1-\alpha)(\alpha(T+2)-1)\right]^2}
$$
\end{theorem}

\begin{proof}
Let us write the  function \eqref{Fz} in the form
 $$
 \Phi(z)=(1-\gamma)^Tz(\Phi_1(z))^T.
 $$
Because $\Phi(z)$ has no roots in $\mathbb D$ then $\Phi^\prime(z)=0$ means $\Phi_1(z)+Tz\Phi^\prime_1(z)=0.$ Since 
$$
\Phi^\prime_1(z)=\frac{q^\prime(z)(1-\gamma p(z))+\gamma q(z)p^\prime(z)}{(1-\gamma p(z))^2}
$$ 
then $\Phi^\prime(z)=0$ means $q(z)+Tzq^\prime(z)+\gamma(Tzq(z)p^\prime(z)-Tzq^\prime(z)p(z)-q(z)p(z))=0.$

In the case $q(z)=\frac{T+1}{T+2}+\frac1{T+2}z, \, p(z)=\alpha z+(1-\alpha)z^2,\; \alpha\in[0,1].$ 
Denote 
$$
\Phi_2(z)\equiv z^3\gamma(T-1)(1-\alpha)+z^2\gamma(T(2T+1)(1-\alpha)-1)+
$$
$$
z(T+1)(\gamma \alpha(T-1)+1)+T+1
$$
then
\begin{equation}\label{F'}
\Phi^\prime(z)=0\Leftrightarrow \Phi_2(z)=0.
\end{equation}
The case of interest i $a\le\frac23$ and $T>1.$ In this case  all coefficients of the polynomial in $\Phi_2(z)$ are positive.

Beside that
$$
\gamma(T(2T+1)(1-\alpha )-1)\cdot(T+1)(\gamma \alpha (T+1)+1)-\gamma(T-1)(1-\alpha )(T+1)\ge
$$
$$ 
\gamma(T+1)((T(2T+1)(1-\alpha )-1)-(T-1)(1-\alpha ))=
$$
$$
\gamma(T+1)(2T^2(1-\alpha )-\alpha )\ge\gamma(T+1)(T^2-1)>0.
$$
By the Gauss-Hurwitz criterion all roots of the equation \eqref{F'} have negative real part.\\

Let us compute $\Phi_2(-1)=T^2\gamma(2-3\alpha).$ Under our assumptions $\Phi_2(-1)>0.$ That means that for $\gamma=0$ the polynomial $\Phi_2(z)$ has a unique root $z_0=-1.$ With $\gamma$ increasing that  real root is moving along the real axis being less then -1 and two additional roots emerge  from $-\infty.$  With further increase  of $\gamma$ all three roots will be continuously depending on the parameter $\gamma.$  Let us find a critical value $\gamma_0$ that provide at least one of the root to be on the boundary of $\mathbb D.$ That might happen only when the pair of the conjugate roots appear on the boundary. In that case the real root necessary to be 
$$
z_{\gamma_0}=-\frac{T+1}{\gamma_0(T-1)(1-\alpha)}
$$
by the Vieta theorem. That implies that $\Phi_2(z_{\gamma_0})=0.$ Thus the parameter $\gamma_0$ has to satisfy the equation
$$
\gamma_0^2-\frac{(T+1)((2T^2+1)(1-\alpha)-1)}{(T-1)^2(1-\alpha)(\alpha(T+2)-1))}\gamma_0+
\frac{(T+1)^2}{(T-1)^2(1-\alpha)(\alpha(T+2)-1)}=0.
$$
Let us note that $(2T^2+1)(1-\alpha)-1>0$ if $\alpha\le2/3$ and $T>1,$ and $\alpha(T+2)-1>0$ for $\alpha>\frac1{T+2}.$

Let us find the discriminant of this equation
$$
\mathcal D=\frac{(T+1)^2\left(((2T^2+1)(1-\alpha)-1)^2-4(T-1)^2(1-\alpha)(\alpha(T+2)-1)\right)}{\left[(T-1)^2(1-\alpha)(\alpha(T+2)-1)\right]^2}
$$
or
$$
\mathcal D=\frac{(T+1)^2}{\left[(T-1)^2(1-\alpha)(\alpha(T+2)-1)\right]^2}\left[ (2-3\alpha)^2+\right.
$$
$$
\left. 4T(T-1)(1-\alpha)[2-3\alpha+T((T-1)(1-\alpha)+2-3\alpha)]\right].
$$
Thus, miracally $\mathcal D$ is positive for $T>1$ and $\alpha\le2/3.$\\

If $\alpha>\frac1{T+2}$ then the minimal root of the equation is
$$
\gamma_0=\frac12\left( \frac{(T+1)((2T^2+1)(1-\alpha)-1)}{(T-1)^2(1-\alpha)(\alpha(T+2)-1)} -\sqrt{\mathcal D}\right).
$$
If $\alpha=\frac1{T+2}$ then
$$
\gamma_0= \frac{(T+1)(T+2)}{(2T^2+1)(T+1)-(T+2)}.
$$
If $\alpha<\frac1{T+2}$ then the minimal root of the equation is
$$
\gamma_0=\frac12\left( \frac{(T+1)((2T^2+1)(1-\alpha)-1)}{(T-1)^2(1-\alpha)(\alpha(T+2)-1)} +\sqrt{\mathcal D}\right).
$$
\end{proof}

The figure below displays the image of the unit disc under the map $\Phi(z)$ for $\alpha=\frac1{T+2}$ (red), $\alpha<\frac1{T+2}$ (green) and $\alpha>\frac1{T+2}$ (blue) for  $\gamma=\gamma_0.$

\centerline{
\includegraphics[scale=0.25]{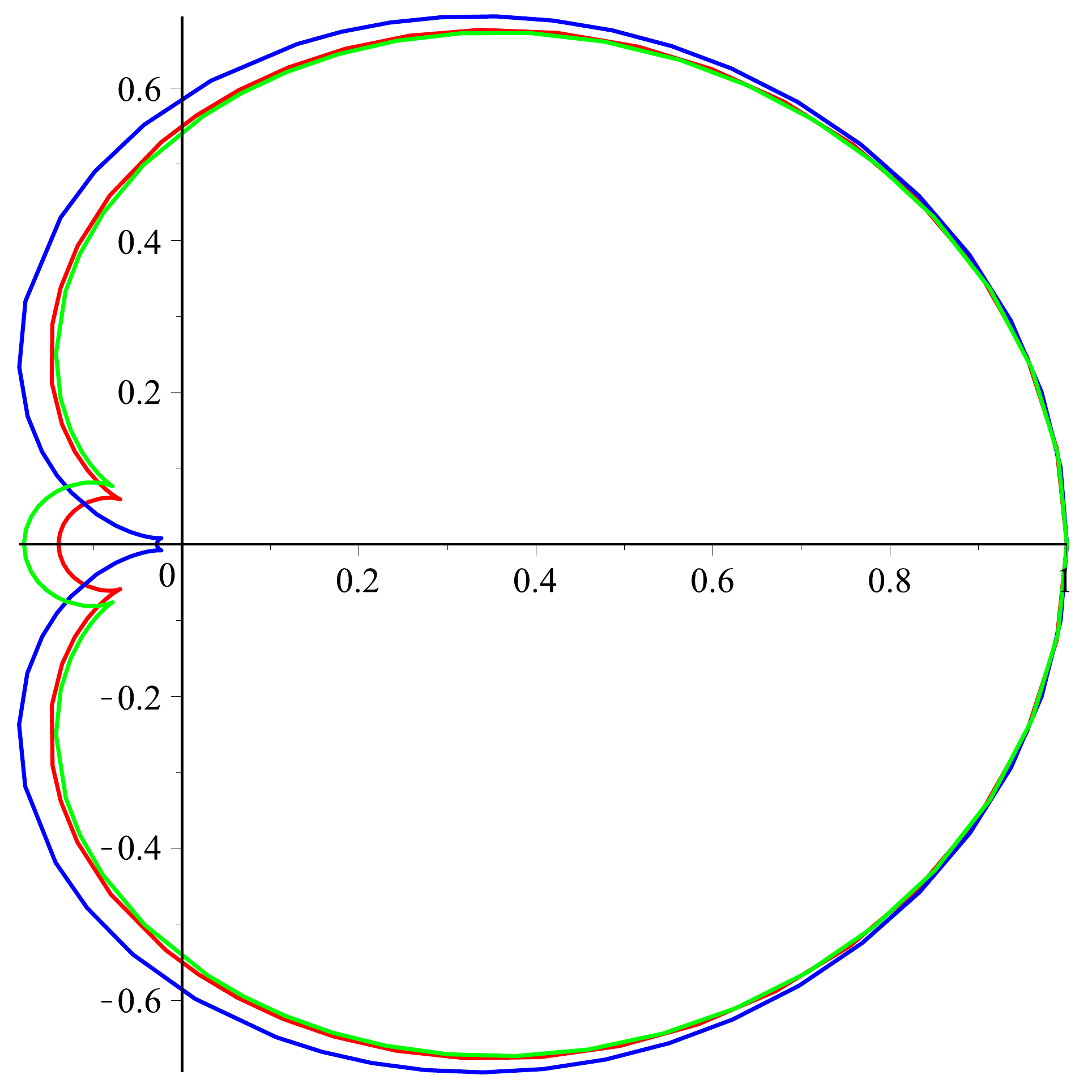}}

The images of $(\mathbb C \backslash \Phi(\mathbb D))^*$ is displayed below.

\centerline{
\includegraphics[scale=0.25]{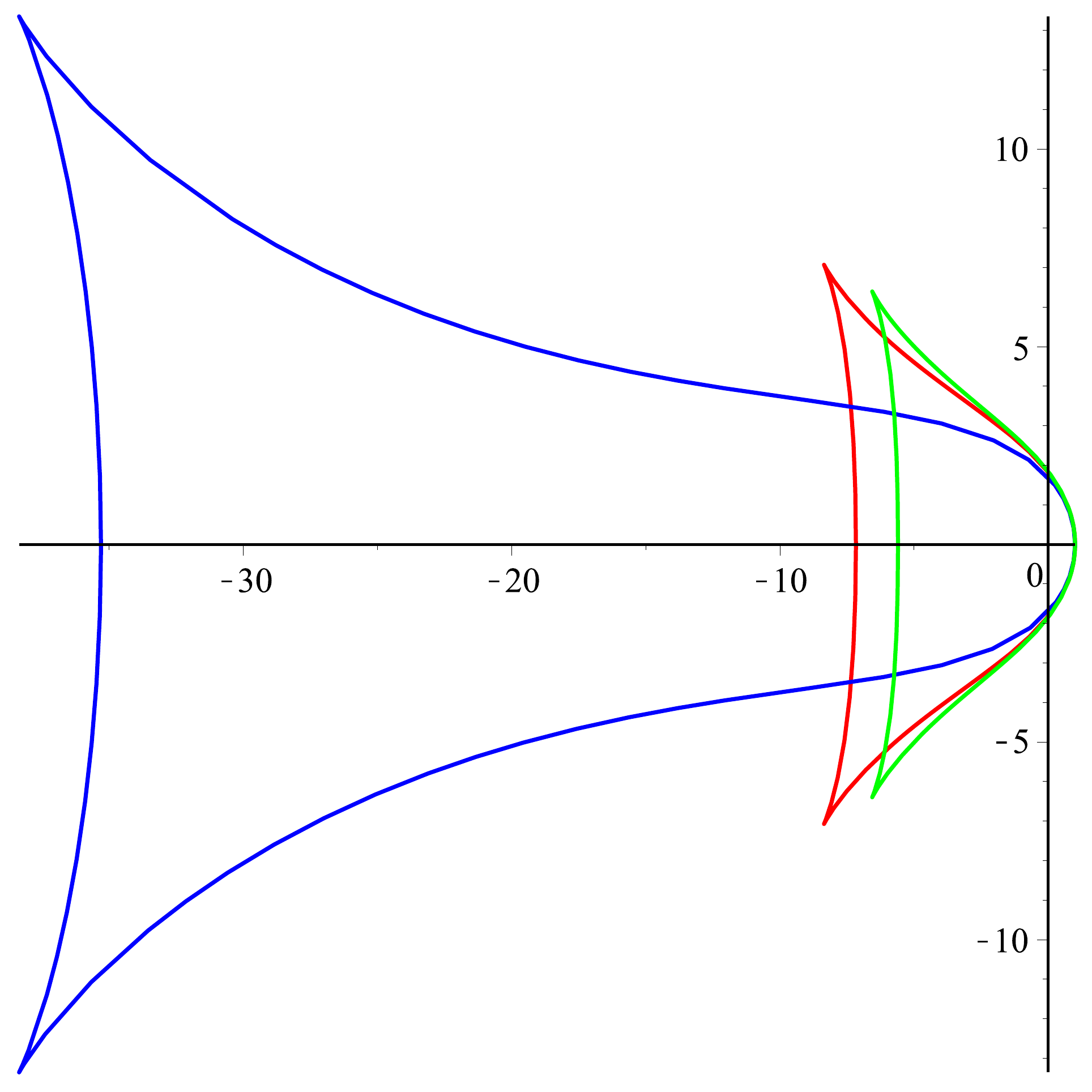}}

The numeric computations demonstrates that for all $T>1$ the value $F(-1)$ for $\alpha=\frac1{2}$
is strictly larger then for $\alpha=\frac1{T+2}.$ Therefore we will be focused on the case $\alpha>\frac1{T+2}.$\\

Let us compute the asymptotic of $\gamma_0$ for large $T$
$$
\gamma_0\sim \frac1{2(1-\alpha)}\frac1T+o\left(\frac1T\right)
$$
and
$$
F(-1)\sim \frac{\left(1-\frac1{2(1-\alpha)}\frac1T\right)^T \left(\frac T{T+2}\right)^T }{\left(1-\frac1{2(1-\alpha)}\frac1{T(1-2\alpha)}\right)^T}\sim\frac {e^{-\frac1{2(1-\alpha)}}e^{-2}}{e^{-\frac{1-2\alpha}{2(1-\alpha)}}}=
e^{-\frac{\alpha}{1-\alpha}-2}.
$$
That implies that the boundary for the multiplier can be chosen as 
$$
\mu^*\sim e^{\frac{\alpha}{1-\alpha}+2}<e^4.
$$

\section{Discussion on the results}
Thus, we have established an algorithm for computing the values of $c^*$ and $\gamma_0,$ that depends on $\alpha$ and T and such that for $\gamma\in [0,c^*)$ the function $\Phi(z)$ is typically real while for $\gamma\in [0,\gamma_0)$ is locally univalent.

Let us compare these values for $\alpha=2/3$ and for $\alpha=1/2.$ Parallel let us compute the corresponding values for $\mu^*.$ For comparison let us show the value of the quantity $(T/(T-2))^T$ which is a boundary of the maximal interval by Morg\"ul's method.

For $a=2/3$\\

\centerline{
\begin{tabular} {|p{0.5in} || p{0.5in} | p{0.5in}| p{0.5in}|p{0.5in}|}\hline 
T & $\gamma_0$  & $c^*$  &$\mu^*$&$(\frac T{T-2})^T$ \\ \hline 
3 & 0.85714 & 0.42254 & 35.698 & 27\\ \hline 
4 &0.55556 &0.25424 & 22.661 & 16\\ \hline 
5  & 0.40909& 0.17797& 19.113& 12.56\\ \hline
6   & 0.32308& 0.13548& 17.543& 11.39\\ \hline
7    & 0.26667 & 0.10876& 16.685& 10.54\\ \hline
8     & 0.22689& 0.09054& 16.156& 9.99\\ \hline
9      & 0.19737& 0.07739&15.803& 9.60\\ \hline
10       & 0.17460& 0.06748&15.555&9.31 \\ \hline
\end{tabular}
}
\bigskip

For $a=1/2.$\\

\centerline{
\begin{tabular} {|p{0.5in} || p{0.5in} | p{0.5in}| p{0.5in}|}\hline 
T & $\gamma_0$  & $c^*$ &$\mu^*$ \\ \hline  
3 & 0.49194 & 0.64145 & 101.640 \\ \hline 
4 & 0.33567& 0.43635 & 50.157\\ \hline 
  5& 0.25365& 0.32474 & 38.309\\ \hline
   6& 0.20343&  0.25703& 33.404\\ \hline
    7& 0.16962&  0.21193& 30.766\\ \hline
     8& 0.14536&  0.17993& 29.140\\ \hline
      9& 0.12711&  0.15613& 28.046\\ \hline
       10& 0.11291&  0.13777& 27.264\\ \hline
\end{tabular}
}
\bigskip

The analysis of the obtained values leads to the following conclusions: for $a=2/3$ the typical realness is violated earlier then the local univalency. At the same time the quantity $\mu^*$ is larger then in Morgul's method about 1.4 times.; for $a=1/2$ the local univalency is violated earlier then the typical realness and $\mu^*$ is larger then in Morgul's case about 3 times. \\

Now, to compare how changes the region $W$ with loosing of univalency (local univalency) let us consider the following problem: for each T find the value $\alpha^*$ such that for this value of $\alpha$ one gets $c_0=\gamma^*.$ Let us show the some numeric solutions to this problem and for the given values $\alpha^*$ let us compute $\mu^*.$\\

\centerline{
\begin{tabular} {|p{0.5in} || p{0.5in} | p{0.5in}| p{0.5in}|}\hline 
T & $\alpha^*$  & $\gamma^*$  &$\mu^*$\\ \hline 
3 & 0.5731 & 0.6932 & 95.521 \\ \hline 
4 &0.5697 &0.4017 & 49.122\\ \hline 
5  & 0.5655& 0.2979& 38.186\\ \hline
6   & 0.5619& 0.23588& 33.511\\ \hline
7    & 0.5589 & 0.1946& 30.967\\ \hline
8     & 0.5565& 0.1655& 29.385\\ \hline
9      & 0.5545& 0.1438&28.313 \\ \hline
10       & 0.5527& 0.1271&27.542 \\ \hline
\end{tabular}
}
\bigskip
As one can see from the table for small T the "optimal" (green) value for $\mu^*$ is even slightly worse then for $\alpha=1/2.$

\centerline{
\includegraphics[scale=0.25]{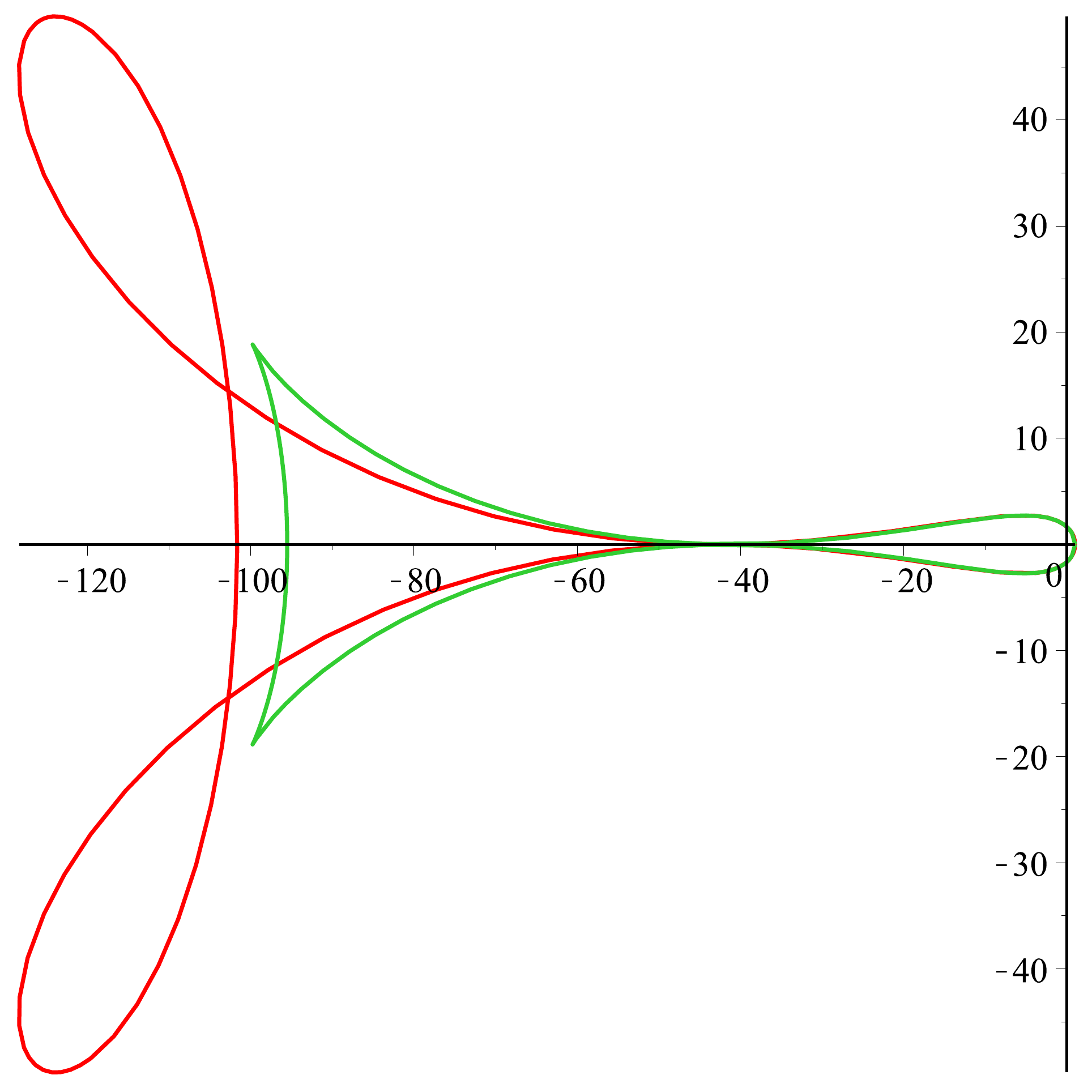}}
\centerline{Inverse image of $\partial \mathbb D$ for $\alpha=1/2$ (red) and for $\alpha^*$ (green), $T=3; a=0.5730949$ }

\centerline{
\includegraphics[scale=0.25]{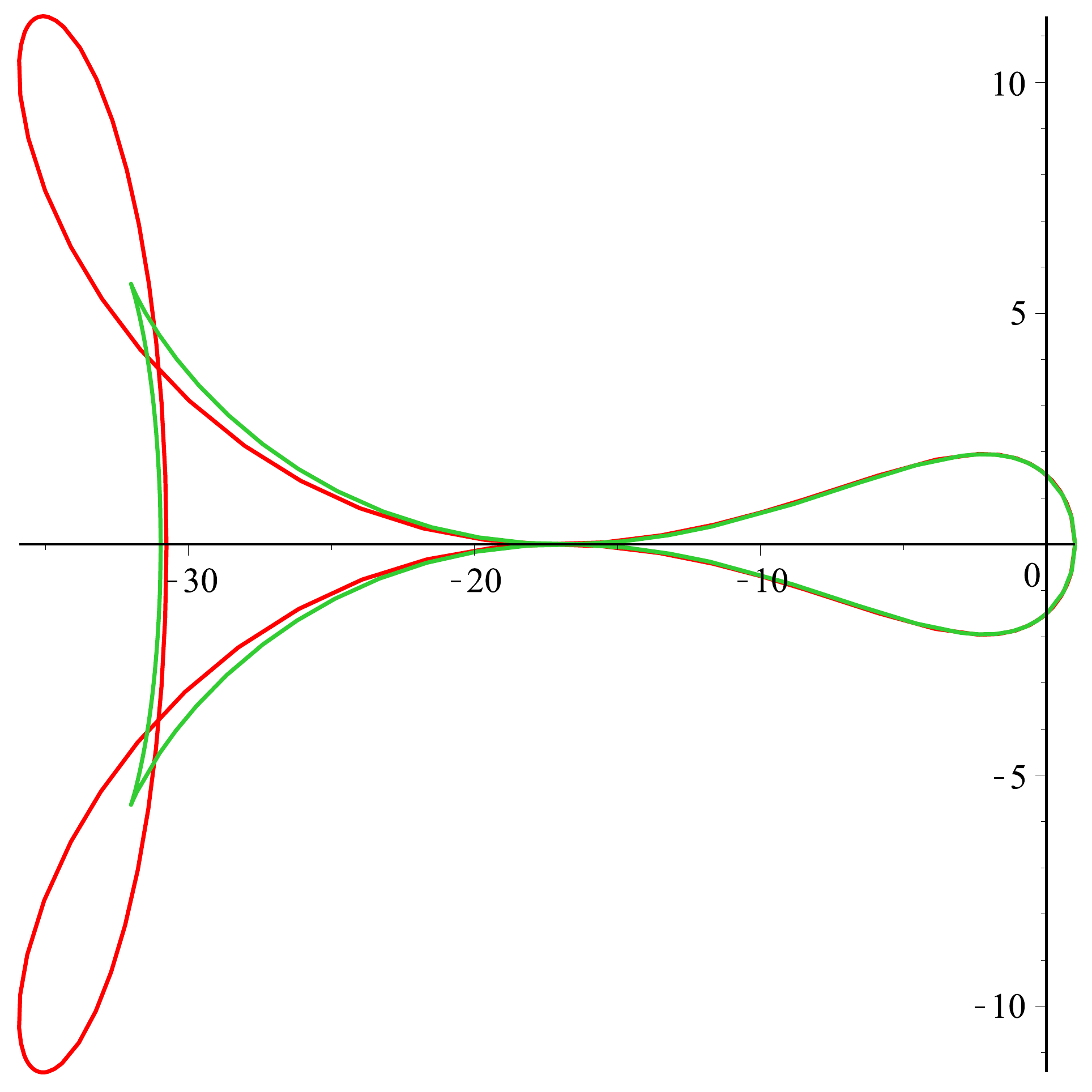}}

\centerline{Inverse image of $\partial \mathbb D$ for $\alpha=1/2$ (red) and for $\alpha^*$ (green), $T= 7; a= 0.558936$ }

\bigskip

Note, that the interior of the read curves without loops is an inverse image of the unit disc $\mathbb D.$ We measure the linear size of the multipliers location region along the horizontal axe and the difference is evident from the pictures.\\

Further analysis of the quantities $\alpha^*$ and $\mu^*$ for large T indicates that $\alpha^*\approx \frac{\pi-2}{\pi-1}$ and 
$\mu^*\approx e^\pi.$ Thus, the method suggested in the article allows to stabilize the cycles with multipliers in the regions of linear size larger then in the Morgul's case in approximately $e^{\pi-2}$ times, i.e. more then 3 times.

Additional analysis indicates that for $\alpha=1/2$ as $c^*$ one can take the value $\frac{1.13}{T-1}.$ The following picture displays the graphs of 
$$
\left(\frac{T+2}T\cdot \frac1{1-\frac{1.13}{T-1}}\right)^T\sim e^{3.13} \qquad \mbox{in black and yellow}
$$
$$
\left(\frac T{T-2}\right)^T \sim e^{2}\qquad\mbox{in blue and yellow}
$$
\centerline{
\includegraphics[scale=0.35]{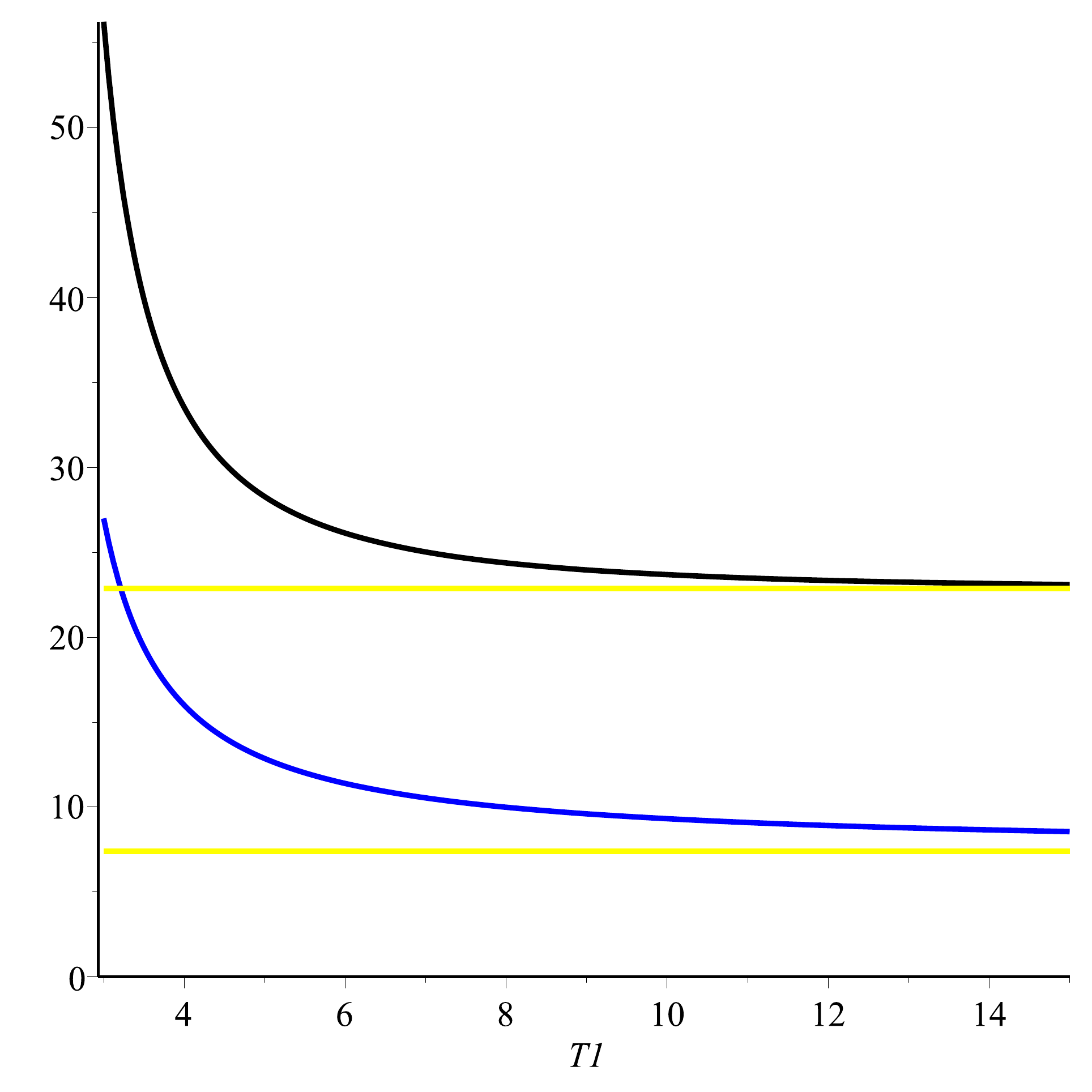}}

\bigskip

The above diagram illustrates that our modification of Morg\'ul's appraoach increases initial region of multipliers locations more then in 3 times.
\bigskip


\section{Numeric computations}

\subsection{Logistic map}
The simplest example for the illustartion of our approach is logistic map $x_{n+1}=\mu x_n(1-x_n)$. For $\mu=4$ Morg\"ul's method allows to find the cylcles of the length 1,2,3 \cite{9} while our approach works for 4,5 and 6. For some other $\mu<4$ Morg\"ul was able to find a cycle of length 20 \cite{10}, our approach allows to find 20-cycles for a wider range of $\mu.$\\

Much better effect of application of our method for stabilization of the cycles is obtained for vector systems. Below we demostrates the application of our method to some famouse maps from nonlinear physics.

\subsection{Elhaj-Sprott map} 
$$
x_{n+1}=1-4\sin(x_n)+0.9y_n,\quad y_{n+1}= x_n.
$$
Below the grey points are of chaotic attarctor while black points are points of the cycle. These cycles are non-stable in the open loop system and locally asymptotically stable in the close loop system.

\centerline{
\includegraphics[scale=0.25]{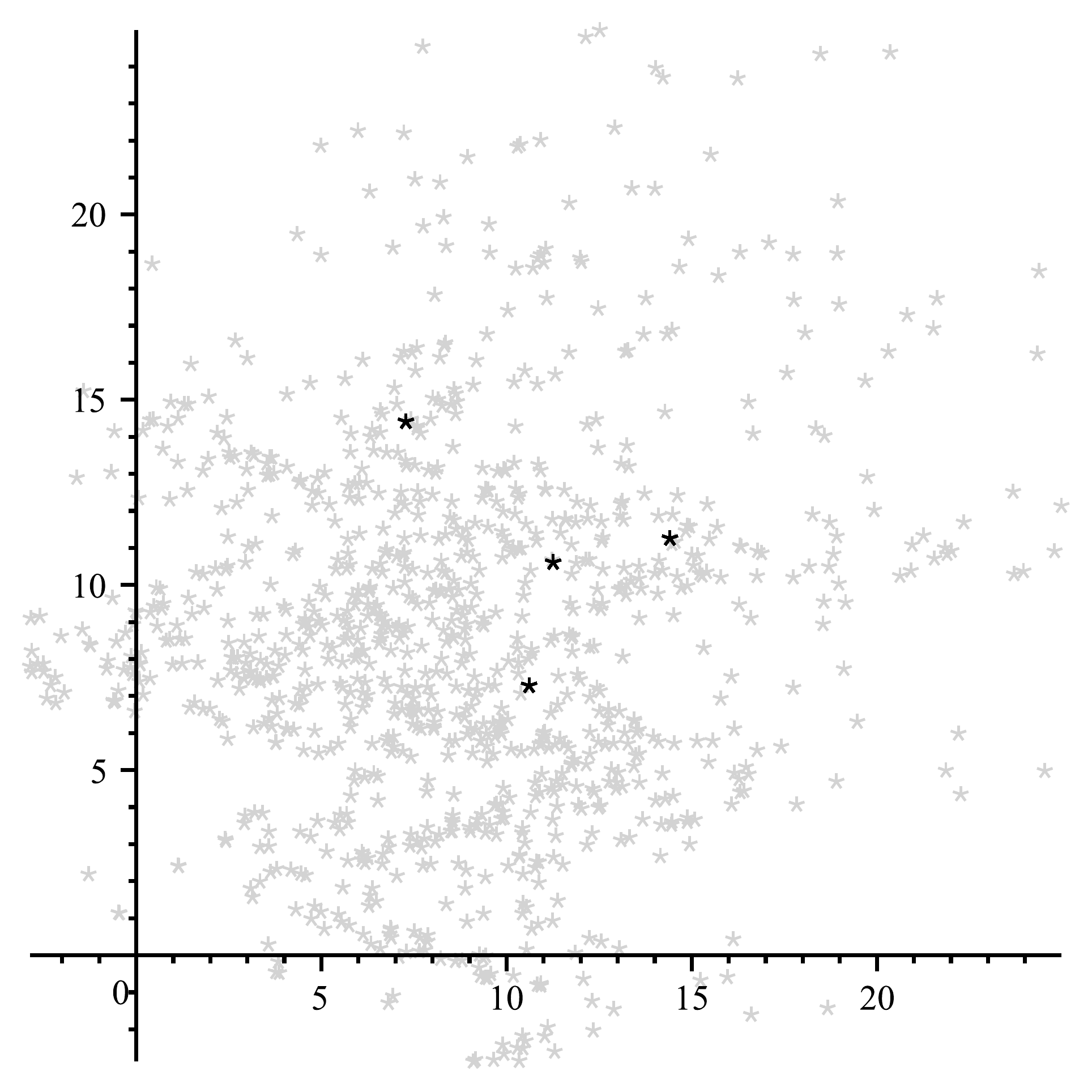}}

\centerline{The plot of 4-cycle for Elhaj-Sprott map  $\gamma=0.25$}
\medskip

\subsection{Ikeda map} 
$$
\begin{cases}
x_{n+1}=&1+0.9\left(x_n\cos\left(0.4-\frac6{1+x_n^2+y_n^2}\right) - y_n \sin\left(0.4-\frac6{1+x_n^2+y_n^2}\right) \right),
\\
y_{n+1}=& 0.9\left(x_n\sin\left(0.4-\frac6{1+x_n^2+y_n^2}\right) + y_n \cos\left(0.4-\frac6{1+x_n^2+y_n^2}\right) \right),
\end{cases}
$$


\centerline{
\includegraphics[scale=0.35]{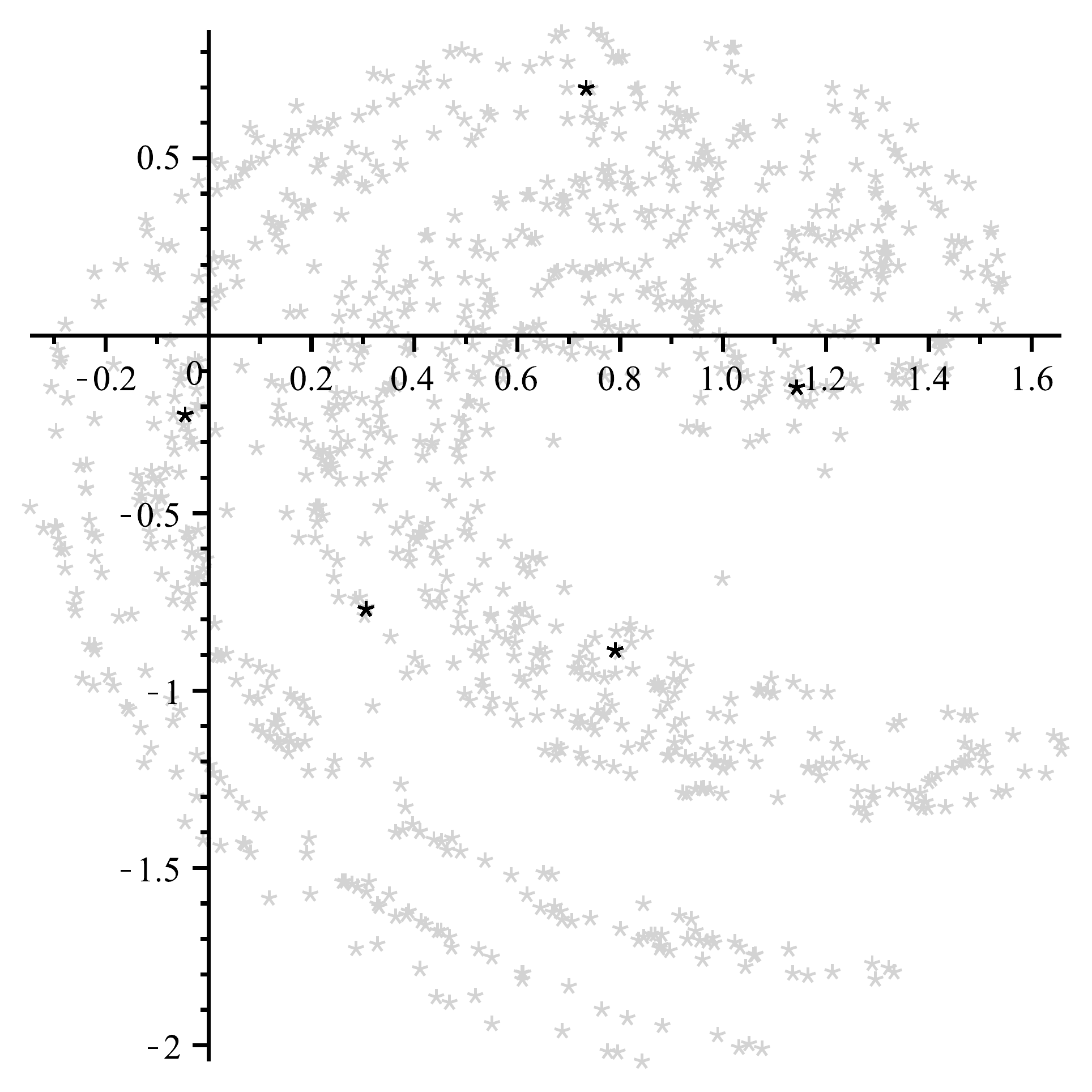}}

\centerline{The plot of 5-cycle for Ikeda map  $\gamma=0.3247$}

\medskip

\subsection{Holms cubic map}
$$
\begin{cases}
x_{n+1}=y_n\\
y_{n+1}=-0.2x+2.77y-y^3.
\end{cases}
$$

\centerline{
\includegraphics[scale=0.35]{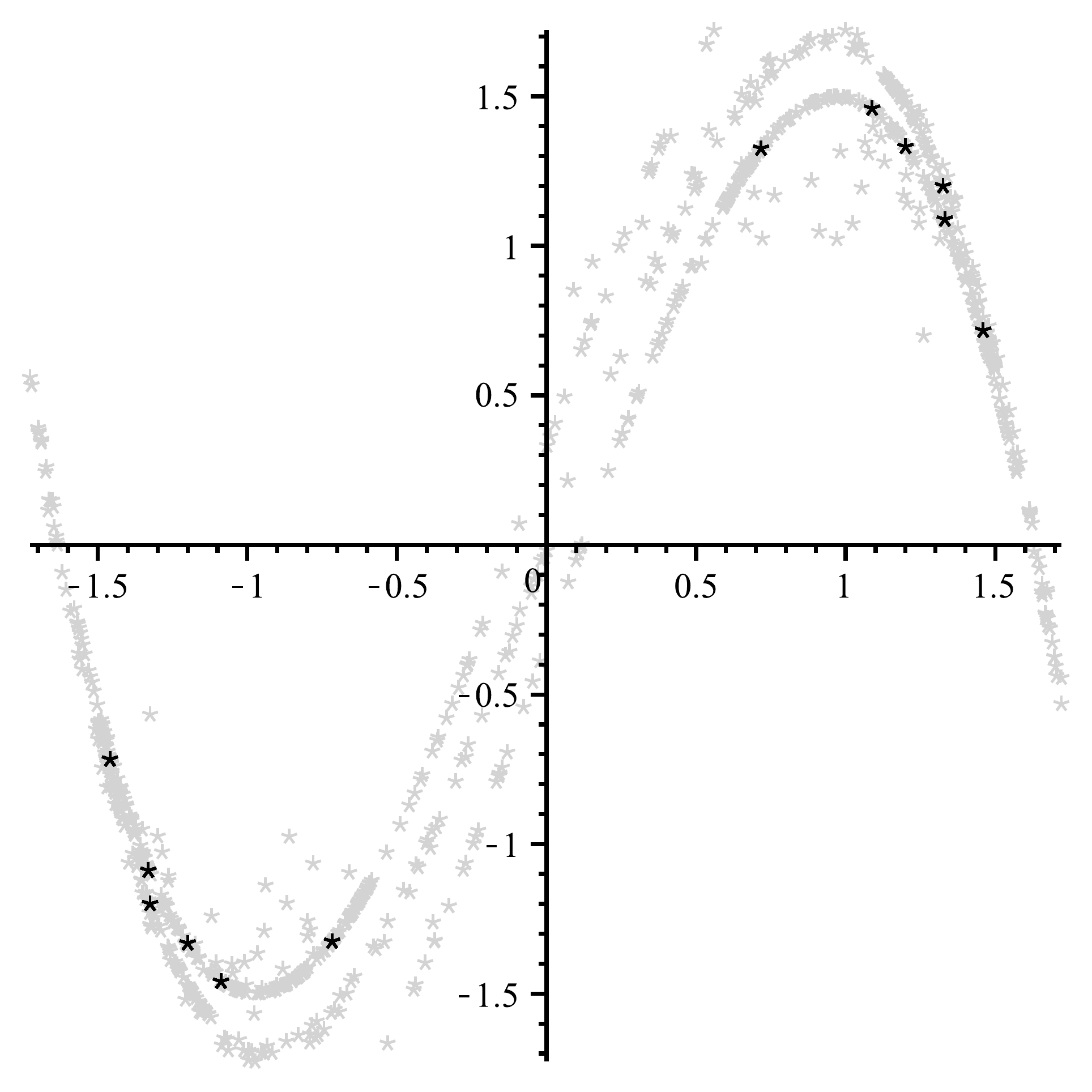}}

\centerline{The plot of 6-cycle for Holms cubic map  $\gamma=0.257$}

\subsection{Henon map} 
$$
x_{n+1}=1-\alpha x_n^2+y_n,\quad y_{n+1}=\beta x_n,\quad \alpha=1.4,\; \beta=0.3. 
$$

\centerline{
\includegraphics[scale=0.25]{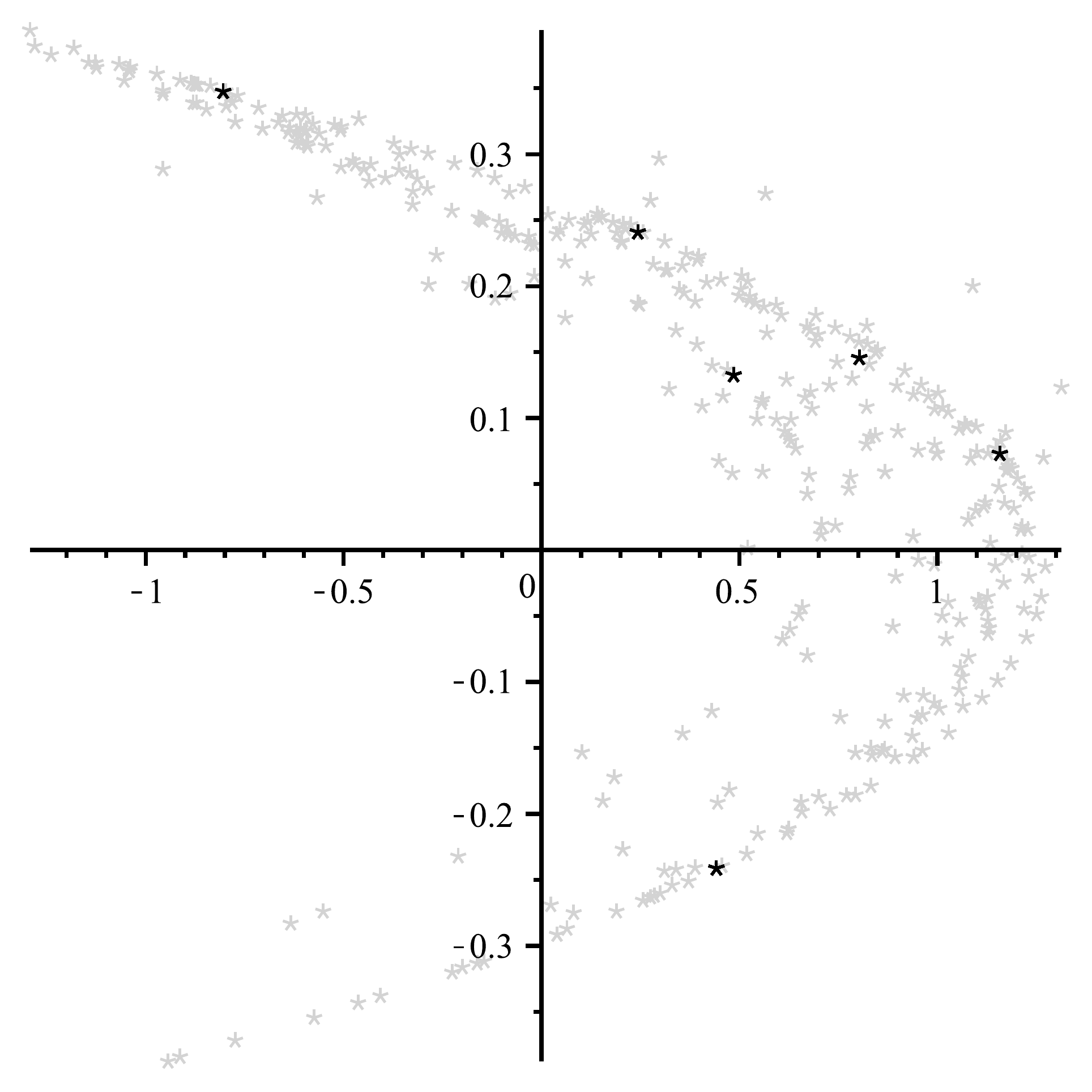}}

\centerline{The plot of 6-cycle for Henon map  $\gamma=0.257$}

\subsection{Lozi map}  Finally, we were able to find 8 cycle in Lozi map 
$$
x_{n+1}=1-\alpha|x_n|+y_n,\quad y_{n+1}=\beta x_n,\quad \alpha=1.4,\; \beta=0.3. 
$$ 

\centerline{
\includegraphics[scale=0.25]{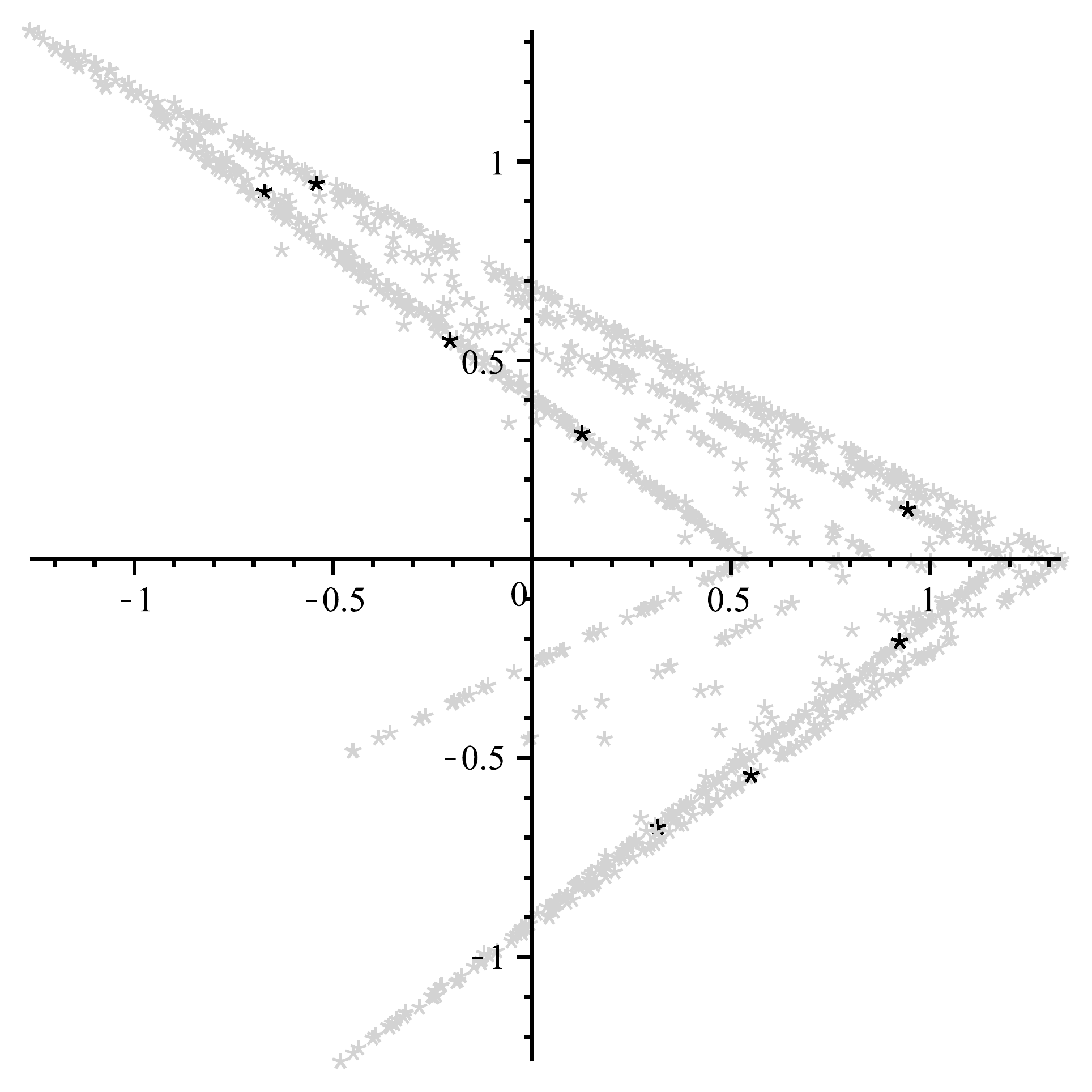}}

\centerline{The plot of 8-cycle for Lozi map  $\gamma=0.179$ }
\medskip

\section{Conclusion}

The current paper address the problem of finding cycles with small multipliers, which is a typical situation for short cycles, say $T\le 20$. For $T>>10$ the region of localization of the multipliers are larger and increases with increase of $T.$

To find short cycles we suggests a {\it simple} schedule which is based on the previous stages. For long cycles are needed more sophisticated schemes that we develop as well. 

Our schedule uses semilinear control \eqref{(6)} with two delays. The advantage of such control is a linear choise for the polynomial $q(z)$ in the denominator of the rational function \eqref{typ}. The number of delays depends on the size of the multipliers. Unfortunately, in case of three and more delays the choise for the polynomial $q(z)$ is not obvious at all. Thus in case $T>>2$ the choise of the control coefficient is difficult to state. So, the choice of the coefficient is rather an art - it depends on skills of a researcher.

We suggested specific algorithms in  \cite{autom,Har,HA}.

To compute the parameters of control in the current article we construct an auxiliary rational function and search the set of the exceptional values of that function in the unit disc.

We have obtained the necessary and sufficient conditions for the typical realness and local univalency. As a conjecture we can suggest that in our partial case the typical realnes and local univalency imply univalency.
Further, let us mention that beside the DFC scheme one can consider mixing. For both of those schemes the auxiliary rational function are the same, the control parameters are coincided. However, the suggested schedule 
$$
x_{n+1}=(1-\gamma)f\left(\frac{T+1}{T+2}x_n+\frac1{T+2}x_{n-T}\right)+\frac\gamma2 (x_{n-T+1}+x_{n-2T+1}),\quad T>1.
$$
is more convenient because on each step it requires computation only one value of function.

\section{Acknowledgment.} The authors would like to thank  Emil Icob, Paul Hagelstein and Alexey Solyanik for interesting discussions and for the help in preparation of the manuscript.


\begin{thebibliography}{12}

\bibitem{1} Jackson E.A. Perspectives of Nonlinear Dinamics. Vol. I, II, - Cambridge Univ. Press, Cambridge, 1980, 1990 Chaos II, ed. Hao Bai-Lin. World Sci., (1990)

\bibitem{2} Ott E., Grebodgi C., Yorke J.A. Controlling chaos. Phys. Rev. Lett. 64, 1196-1199 (1990)

\bibitem{3} Chen G., Dong X. From chaos to order: Methodologies, Perspectives and Application. World Scientific, Singapore (1999)

\bibitem{4} Andrievsky B. R., Fradkov A. L. Control of Chaos: Methods and Applications. I. Methods, Avtomat. i Telemekh., (2003), no. 5, 3–45

\bibitem{5} Pyragas K. Continuous control of chaos by self controlling feedback. Phys. Rev. Lett. A 170, 421–428 (1992)

\bibitem{6} Vieira de S.M., Lichtenberg A.J. Controlling chaos using nonlinear feedback with delay. Phys. Rev. E 54, 1200-1207 (1996)

\bibitem{7} Dmitrishin D. and Khamitova A. Methods of harmonic analysis in nonlinear dynamics, Comptes Rendus Mathematique, Volume 351, Issues 9-10, 367 - 370 (2013)

\bibitem{8} Dmitrishin D., Skrinnik I., Stokolos A. From chaos to order through mixing, arXiv:1607.05493 [nlin.CD] (2016)

\bibitem{9} Morgul O. On the stability of delayed feedback controllers. Phys. Lett. A. 314, 278-285 (2003)

\bibitem{10} Morgul O. Further stability results for a generalization of delayed feedback control, Nonlinear Dynamics, 1-8 (2012)

\bibitem{13} Khalil H.K. Nonlinear Systems, 3rd edn. Prentice-Hall, Upper Saddle River (2002)

\bibitem{14} Elaydi S. N., An Introduction to Difference Equations, Springer-Verlag, New York, 1996

\bibitem{15} Solyanik A. A-Stabilization and the ranges of complex polynomials on the unit disk, arXiv:1701.04784v [math.NA] (2017)

\bibitem{arxiv}  Dmitrishin D.,  Khamitova A. and Stokolos A., On the generalized linear and non-linear DFC in non-linear dynamics, arXiv:1407.6488, (2015).

\bibitem{Har} Dmitrishin D., Franzheva E.,Skrinnik I. and Stokolos A., Generalization of nonlinear control for nonlinear discrete systems.  Bulletin of NTU "KhPI" 2017. Series: System Analysis, Control and information technology, No 28 (1250), pp3-18. ISSN 2079-0023 (in Russian, English summary).

\bibitem{HA}   Dmitrishin D., Khamitova A., Stokolos A., Fej\'er Polynomials and Chaos,  Special Functions, Partial Differential Equations, and Harmonic Analysis, Springer 2014.pp 49-75

\bibitem{HA1} Dmitirshin D., Khamitova A., Tohaneanu M. and Stokolos  A., Finding Cycles in Nonlinear Autonomous Discrete Dynamical Systems, Harmonic Analysis, Partial Differential Equations, Banach Spaces, and Operator Theory (Volume 2), Springer 2017, pp 199-237.

\bibitem{autom}  Dmitrishin D., Franzheva E. and Skrinnik I.,
Methods of Geometric Complex Analysis and the periodic point problem, arxiv, 2017.

\bibitem{PSz}   P\'olya G. and Szeg\"o G., Problems and theorems in analysis, Springer 1972.

\bibitem{Ru} Cordova A.Y. and Ruscheweyh S., On maximal ranges of polynomial spaces in the unit disk, Constructive Approximation, 1989, Volume 5, Issue 1, pp 309-327

\end{thebibliography}
\end{document}